\documentclass[preprint,12pt]{elsarticle}

\usepackage{algorithm}
\usepackage{algorithmic}
\usepackage{amsmath,amssymb,amsfonts,graphics,
latexsym, amscd, amsfonts, epsfig, eepic,epic,psfrag,setspace,eufrak}
\usepackage[all]{xypic}

\usepackage{amssymb}
\usepackage{amsthm}

\theoremstyle{plain}
\newtheorem{theorem}{Theorem}
\newtheorem{corollary}{Corollary}

\newtheorem{lemma}{Lemma}
\newtheorem{proposition}{Proposition}
\theoremstyle{definition}
\newtheorem{definition}{Definition}

\theoremstyle{remark}

\journal{Mathematical Biosciences}

\begin{document}

\begin{frontmatter}



\title{Uniform generation of RNA pseudoknot structures with genus filtration}


\author[dk]{Fenix W.D.\ Huang}
\ead{fenixprotoss@gmail.com}

\author[dk,kl]{Markus E.\ Nebel}
\ead{nebel@informatik.uni-kl.de}

\author[dk]{Christian M.\ Reidys\corref{cor1}}
\ead{duck@santafe.edu}

\cortext[cor1]{Corresponding author}

\address[dk]{Department of Mathematic and Computer science, University of
    Southern Denmark, Campusvej 55, DK-5230 Odense M, Denmark}
\address[kl]{Department of Computer Science, University of Kaiserslautern, Germany}

\begin{abstract}
In this paper we present a sampling framework for RNA structures of fixed topological
genus. We introduce a novel, linear time, uniform sampling algorithm for RNA structures
of fixed topological genus $g$, for arbitrary $g>0$.
Furthermore we develop a linear time sampling algorithm for RNA structures of
fixed topological genus $g$ that are weighted by a simplified, loop-based energy
functional. For this process the partition function of the energy functional has
to be computed once, which has $O(n^2)$ time complexity.
\end{abstract}


\end{frontmatter}

\section{Introduction}

Pseudoknots have long been known as important structural elements in RNA
\cite{Westhof:92a}. These cross-serial interactions between RNA nucleotides
are functionally important in tRNAs, RNaseP \cite{Loria:96a}, telomerase
RNA \cite{Staple:05}, and ribosomal RNAs \cite{Konings:95a}.
Pseudoknots in plant virus RNAs mimic tRNA structures, and {\it in vitro}
selection experiments have produced pseudoknotted RNA families that bind
to the HIV-1 reverse transcriptase \cite{Tuerk:92}. Import general mechanisms,
such as ribosomal frame shifting, are dependent upon pseudoknots \cite{Chamorro:91a}.

Lyngs{\o} {\it et al.} \cite{Lynsoe:00} have shown that the prediction of general
RNA pseudoknot structures is NP-complete. Thus, in order to provide prediction tools of
feasible time complexity one frequently sticks to subtle subclasses of pseudoknots
suitable for the dynamic programming paradigm \cite{Rivas:99, Nebel:12}. Alternative approaches to
the prediction of RNA secondary structure (with or without pseudoknots) build on random
sampling of foldings compatible to a given sequence. Here both, the underlying probability
model and the efficiency of the sampling algorithm are crucial for being successful.

In this paper we propose a linear time uniform random sampler for pseudoknotted
RNA structures of given genus which might be considered a promising starting point
for the design of efficient solutions to the structure prediction problem.
Our approach is based on the observation that pseudoknotted RNAs
are in a natural way related to topological surfaces. In fact
pseudoknotted RNA structures can be viewed as drawings on orientable surfaces of
genus $g$, that is by means of the classical classification theorem either on the
sphere (secondary structures) or connected sums of tori (pseudoknotted structures).
Our approach is a natural evolution from Waterman {\it et al.} pioneering work
\cite{Waterman:79a,Nussinov:78,Kleitman:70} on secondary structures.

Secondary structures are coarse grained RNA contact structures, see Figure~\ref{F:structure} (A).
They can be represented as diagrams, i.e.~labeled
graphs over the vertex set $[n]=\{1, \dots, n\}$ with vertex degrees
$\le 3$, represented by drawing its vertices on a horizontal
line and its arcs $(i,j)$ ($i<j$), in the upper half-plane, see
Figure~\ref{F:structure}. We assume the vertices to be connected by the edges $\{i,i+1\}$, $1\le i<n$,
which are not considered arcs (but contribute to a nodes's degree).
Furthermore, vertices and arcs correspond to the nucleotides {\bf A}, {\bf G},
{\bf U} and {\bf C} and Watson-Crick base pairs ({\bf A-U}, {\bf G-C}) or wobble base pairs
({\bf U-G}), respectively.

\begin{figure}[ht]
\begin{center}
\includegraphics[width=0.85\columnwidth]{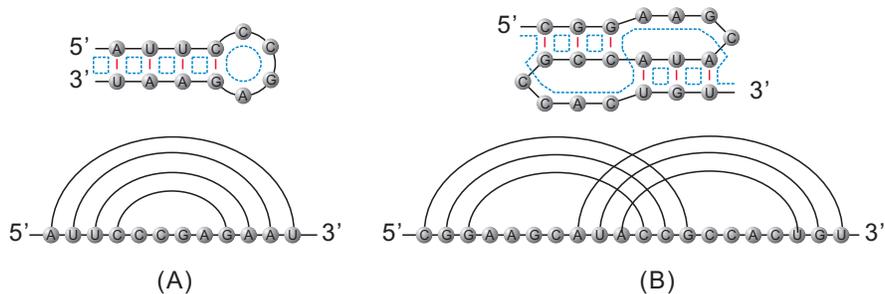}
\end{center}
\caption{\small A secondary structure (A) and a pseudoknot
structure (B) and their diagram representation.
}
\label{F:structure}
\end{figure}

Considering only the Watson-Crick and wobble base pair RNA structures, we set the restriction
that one vertex can only paired with at most another vertex. Let $i<r$, we call arcs
$(i,j)$ and $(r,s)$ crossing if $i<r<j<s$ holds. In this representation a
secondary structure is a diagram without crossing arcs. Otherwise, i.e.~diagrams
with crossings represent pseudoknot structures, see Figure~\ref{F:structure} (B).

In this paper, we present a framework for generating diagrams with crossings, filtered by
topological genus, with uniform probabilities.
The topological filtration of RNA structures has first been proposed by
Penner and Waterman in \cite{Waterman:93} and later, as an application
of the Matrix model \cite{Orland:02}, in \cite{Bon:08}.
The work here however is based on the combinatorial work of Chapuy \cite{Chapuy:11}.

In order to understand how topology enters the picture for RNA molecules we need
to pass from diagrams or contact-graphs to that of topological surfaces. Only the
associated surface carries the key invariants leading to a meaningful filtration
of RNA structures. The mental picture here is to ``thicken'' the edges into (untwisted)
bands and to expand each vertex to a disk as shown in Figure~\ref{F:fat}. This
inflation of edges leads to a fatgraph $\mathbb{D}$ \cite{Loebl:08,Penner:10}.

A fatgraph, sometimes also called also a ``map'', is a graph
equipped with a cyclic ordering of the incident half-edges at each vertex.
Thus, $\mathbb{D}$ refines its underlying graph $D$ insofar as it encodes
the ordering of the ribbons incident on its disks. In fact a fatgraph constitutes to a
cell-complex structure --combinatorial data in a sense-- that have a topological surface
as geometric realization \cite{Massey:69}.

Our sampling process consists of two steps: first we generate a diagram without crossing
arcs and second we lift the topological genus to some fixed $g$. The process has linear
time and is thereby very efficient.

The paper is organized as follows: we first introduce the topological filtration of diagrams.
Then we introduce a genus induction process and finally, we describe and analyze the sampling
processes.

\section{Some basic facts} \label{S:basic}

\subsection{Diagrams}

A diagram is a labeled graph over the vertex set $[n]=\{1, \dots, n\}$ in
which each vertex has degree $\le 3$, represented by drawing its vertices
in a horizontal line. The backbone of a diagram is the sequence of
consecutive integers $(1,\dots,n)$ together with the edges $\{\{i,i+1\}
\mid 1\le i\le n-1\}$. The arcs of a diagram, $(i,j)$, where $i<j$, are
drawn in the upper half-plane. We shall distinguish backbone edges
$\{i,i+1\}$ from arcs $(i,i+1)$, which we refer to as a $1$-arc.
Two arcs $(i,j)$, $(r,s)$, where $i<r$ are crossing if $i<r<j<s$ holds.
The arc $(1,n)$ is called rainbow, see Figure~\ref{F:diagram}.

\begin{figure}[ht]
\begin{center}
\includegraphics[width=0.7\columnwidth]{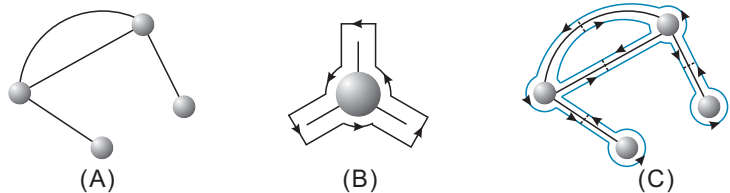}
\end{center}
\caption{\small From graphs to fatgraphs:
(A) A graph with $4$ vertexes and $4$ edges. (B) Inflation of a
vertex. (C) A fatgraph derived from (A) induces a topological surface.
}
\label{F:fat}
\end{figure}

\subsection{Fatgraphs and unicellular maps}
In this section, we discuss the filtration of diagrams by topological
genus. In order to extract topological properties of diagrams those need to be enriched
to fatgraphs. The latter are tantamount to a cell-complex
structures over topological surfaces.
Formally, we make this transition \cite{Reidys:top1} by ``thickening''
the edges of the diagram into (untwisted) bands or ribbons. Furthermore
each vertex is inflated into a disc as shown in Figure~\ref{F:fat} (B).
This inflation of edges and vertices means to replace a set of incident edges
by a sequence of half-edges. This constitutes the fatgraph $\mathbb{D}$
\cite{Loebl:08,Penner:10}.

\begin{figure}[ht]
\begin{center}
\includegraphics[width=0.85\columnwidth]{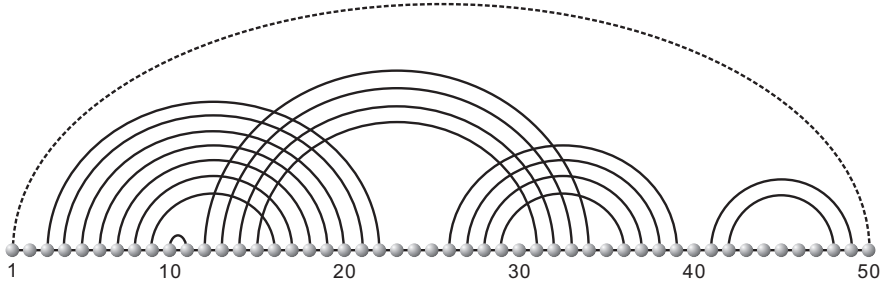}
\end{center}
\caption{\small A diagram over $50$ vertices. The arc $(10,11)$ is
a $1$-arc. The arcs $(3,22)$ and $(12,34)$ are crossing.
The dashed arc $(1,50)$ is the rainbow.
}
\label{F:diagram}
\end{figure}

A fatgraph is thus a graph enriched by a cyclic ordering of the incident
half-edges at each vertex and consists of the following data: a set of
half-edges, $H$, cycles of half-edges as vertices and pairs of half-edges
as edges. Consequently, we have the following definition:

\begin{definition}
A fatgraph is a triple $(H, \sigma, \alpha)$, where $\sigma$ is
the vertex-permutation and $\alpha$ a fixed-point free involution.
\end{definition}

In the following we will deal with orientable fatgraphs\footnote{Here ribbons may also be allowed to twist giving
rise to possibly non-orientable surfaces \cite{Massey:69}.}.
Each ribbon has two boundaries. The first one in counterclockwise
order shall be labeled by an arrowhead, see Figure~\ref{F:fat} (C).

A fatgraph $\mathbb{D}$ exhibits a phenomenon, not present in its
underlying graph $D$. Namely, one can follow the (directed) sides of the
ribbons rotating counterclockwise around the vertices. This gives rise to
$\mathbb{D}$-cycles or boundary components, constructed by following these
directed boundaries from disc to disc. Algebraically, this amounts to form
the permutation $\gamma=\alpha \circ \sigma$.

In the following we consider only diagrams with rainbow. As we shall see,
the rainbow arc provides a canonical first boundary component, which travels
on top of the rainbow arc and around the backbone of the diagram, see
Figure~\ref{F:bc}.

\begin{figure}[ht]
\begin{center}
\includegraphics[width=0.8\columnwidth]{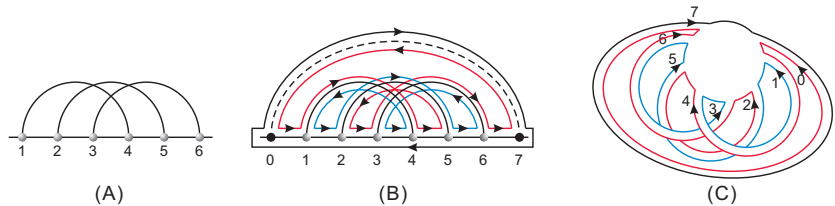}
\end{center}
\caption{\small (A) A diagram.  (B) the fattening of (A) augmented by the
rainbow (0, 7). Here $\sigma=(0,1,2,3,4,5,6,7)$, $\alpha=(0,7)(1,4)(2,5)(3,6)$.
Accordingly $\gamma= \alpha \circ \sigma= (0, 4, 2, 6)(1, 5, 3)(7)$ has two cycles.
(C) Collapsing the backbone into a vertex.
}
\label{F:bc}
\end{figure}

A fatgraph, $\mathbb{D}$, can be viewed as a ``drawing'' on a
certain topological surface.
$\mathbb{D}$ is a $2$-dimensional cell-complex over its geometric
realization, i.e.~a surface without boundary, $X_{\mathbb{D}}$, realized
by identifying all pairs of edges \cite{Massey:69}.
Key invariants of the latter, like Euler characteristic \cite{Massey:69}
\begin{eqnarray}\label{E:euler}
\chi(X_\mathbb{D}) & = & v - e + r,\\
\label{E:genus}
g(X_\mathbb{D}) & = & 1-\frac{1}{2}\chi(X_\mathbb{D}),
\end{eqnarray}
where $v,e,r$ denotes the number of discs, ribbons and boundary components
in $\mathbb{D}$ \cite{Massey:69} are defined combinatorially. However,
equivalence of simplicial and singular homology \cite{Hatcher:02} implies that
these combinatorial invariants are in fact invariants of $X_{\mathbb{D}}$ and thus
topological. This means the surface $X_{\mathbb{D}}$ provides a topological
filtration of fatgraphs.

Since, adding a rainbow or collapsing the backbone of a diagram
does not change the Euler characteristic, the relation between
genus and number of boundary components is solely determined by the number
of arcs in the upper half-plane:
\begin{equation}\label{E:ee}
2-2g-r = 1-n,
\end{equation}
where $n$ is number of arcs and $r$ the number of boundary
components. The latter can be computed easily and allows us therefore
to obtain the genus of the diagram.

\begin{definition}
A unicellular map $\mathfrak{m}$ of size $n$ is a fatgraph
$\mathfrak{m}(n)=(H,\alpha,\sigma)$ in which the
permutation $\alpha\circ\sigma$ is a cycle of length $2n$.
\end{definition}

While unicellular maps are simply particular fatgraphs, they naturally
arise in the context of diagrams, by two observations. First in the
diagram one may collapse the backbone into a single
vertex. Second the mapping
$$
\pi \colon (H,\sigma, \alpha) \ \mapsto \ (H, \alpha\circ \sigma,\alpha),
$$
is evidently a bijection between fatgraphs having one vertex and
unicellular maps, see Figure~\ref{F:dual}.
The mapping is called the {\it Poincar\'{e} dual} and interchanges
boundary components by vertices, preserving topological genus.
In the following, we use $\pi$ to denote the {Poincar\'{e} dual}.

\begin{figure}[ht]
\begin{center}
\includegraphics[width=0.8\columnwidth]{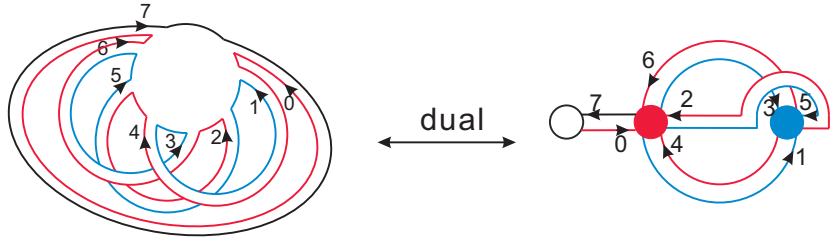}
\end{center}
\caption{\small The Poincar\'{e} dual: we map a fatgraph with $1$
vertex and $3$ boundary components into a fatgraph with $3$ vertexes
and $1$ boundary component.
}
\label{F:dual}
\end{figure}

Given a unicellular map the permutation $\sigma$ and $\gamma$ induces two
linear orders of half-edges
$$
r <_{\gamma} <\gamma(r) <_{\gamma} \dots <_{\gamma} \gamma^{2n-1}(r), \quad
r <_{\sigma} <\sigma(r) <_{\sigma} \dots <_{\sigma} \sigma^{k}(r).
$$
Let $a_1$ and $a_2$ be two distinct half-edges in $\mathfrak{m}$. Then
$a_1<_{\gamma} a_2$ expresses the fact that $a_1$ appears before $a_2$ in
the boundary component $\gamma=\alpha\circ \sigma$.
Suppose two half-edges $a_1$ and $a_2$ belong to the same vertex $v$. Note
that $v$ is effectively a cycle which we assume to originate with the first
half-edge along which one enters $v$ traveling $\gamma$. Then
$a_1<_\sigma a_2$ expresses the fact that $a_1$ appears (counterclockwise)
before $a_2$.

The Poincare-dual maps the rainbow into a distinguished vertex of degree one
and provides thereby a natural origin for the cycle $\gamma$. We call this
vertex the {\it plant}, see Figure~\ref{F:dual}.
Given a unicellular map we call a half-edge the minimum half-edge of a
vertex $v$ if it is the first half-edge via which $\gamma$ visits $v$.

\subsection{Genus induction}

In this section we present a construction of \cite{Chapuy:11}, which plays a
key role for our main result. It consists of two processes: a slicing-map
$\Xi$ and a gluing-map $\Lambda$, which, when restricted to the proper
classes, are inverse to each other.

The slicing process splits a vertex into $(2g+1)$ vertices and thereby
reduces the genus of the map by $g$. Gluing is effectively inverse to
slicing, namely: gluing any $(2g+1)$ vertices in a unicellular map increases
the genus of the map by $g$. Slicing and gluing preserve unicellularity.

\begin{definition}
A half-edge $h$ is an {\it up-step} if $h<_{\gamma} \sigma(h)$, and a
{\it down-step} if $\sigma(h) \le_{\gamma} h$. $h$ is called a {\it trisection}
if $h$ is a down-step and $\sigma(h)$ is not the minimum half-edge of its
respective vertex.
\end{definition}

The number of trisections in a unicellular map of genus $g$ is given by the following
lemma:

\begin{lemma} \label{L:trisection}
\cite{Chapuy:11} Let $\mathfrak{m}$ be a unicellular map of
genus $g$. Then $\mathfrak{m}$ has exactly $2g$ trisections.
\end{lemma}

Given a unicellular map $\overline{\mathfrak{m}}$ and a vertex $\overline{v}$ together with
a trisection $\tau$ contained in $\overline{v}$. Let $a_1$ be the minimum half-edge of
$\overline{v}$. Then set $a_3=\overline{\sigma}(\tau)$ and $a_2$ to be the smallest
half-edge between $a_1$ and $a_3$ (with respect to the order $<_{\overline{\sigma}}$)
such that $a_2>_{\overline{\gamma}}$.

Since $\tau$ is a trisection such an $a_2$ exists. Then we refer to the replacement of
$$
\overline{v}=(a_1, h_2^1,\ldots, h_2^{m_2}, a_2, h_3^1,\ldots,h_3^{m_3},a_3,h_1^{1},\ldots,h_1^{m_1})
$$
by the three vertices where $v_i=(a_i,h_i,^1,\ldots, h_i^{m_i})$, $i=1,2,3$, see
Figure~\ref{F:glue_slice}, as slicing. Slicing produces the unicellular fatgraph
$\overline{\mathfrak{m}}=(H,\overline{\sigma}, \alpha)$.

\begin{figure}[ht]
\begin{center}
\includegraphics[width=0.8\columnwidth]{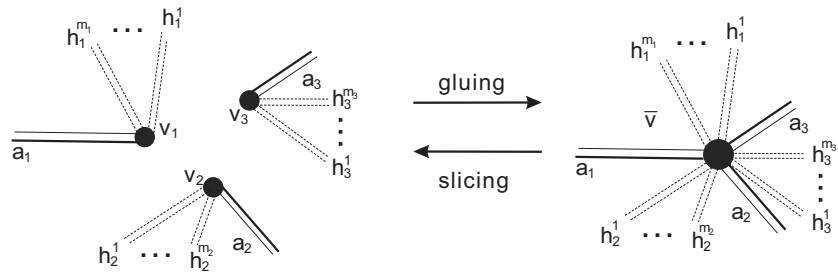}
\end{center}
\caption{\small Illustration of gluing and slicing in a unicellular map.
}
\label{F:glue_slice}
\end{figure}

Conversely, let $\mathfrak{m}$ be a unicellular map and let $a_1$, $a_2$ and $a_3$ be three
half-edges belonging to three distinct vertices, $v_i=(a_i,h_i,^1,\ldots, h_i^{m_i})$ for some
$m_i\ge 0$ and $i=1,2,3$. Furthermore suppose $a_1<_{\gamma}a_2<_{\gamma}a_3$.

Then, replacing the cycles $v_1$, $v_2$ and $v_3$ by the cycle
$$
\overline{v}=(a_1, h_2^1,\ldots, h_2^{m_2}, a_2, h_3^1,\ldots,h_3^{m_3},a_3,h_1^{1},\ldots,h_1^{m_1}),
$$
is referred to as gluing. Gluing produces the unicellular fatgraph $\overline{\mathfrak{m}}=
(H,\overline{\sigma}, \alpha)$, see Figure~\ref{F:glue_slice}, in which the half-edge
$\overline{\sigma}^{-1}(a_3)$ is, by construction, a trisection.

\begin{lemma} \cite{Chapuy:11}
Slicing maps a unicellular map together with a trisection into a unicellular map together with
three labeled vertices. Gluing maps a unicellular map together with three labeled vertices
into a unicellular maps with a trisection.
\end{lemma}

Suppose we slice $(\overline{\mathfrak{m}},\tau)$ into $(\mathfrak{m}, v_1,v_2,v_3)$,
where in $\overline{\mathfrak{m}}$ holds $a_1<_{\overline{\gamma}} a_3 <_{\overline{\gamma}} a_2$.
Then we observe that in $\mathfrak{m}$ $a_1$ remains minimum in its new vertex and so
does $a_2$, because $a_2$ is by definition the minimum half-edge where
$a_3 <_{\overline{\gamma}} a_2$.
However, $a_3$ becomes either the minimum half-edge, or remains a half-edge following a
trisection. This gives rise to {\it two} types of trisections:

\begin{definition}
Let $\overline{\mathfrak{m}}$ be a unicellular map and $\overline{v}$ a vertex
containing a trisection $\tau$.
Slicing $(\overline{\mathfrak{m}},\tau)$ we obtain $(\mathfrak{m},v_1,v_2,v_3)$.
If the minimum half-edge of $v_3$, denoted by $a_3$ is the half-edge
$\overline{\sigma}(\tau)$ in $\overline{\mathfrak{m}}$, we call the trisection
$\tau$ to be of {\it type I} and {\it type II}, otherwise.
\end{definition}
\begin{proposition}
Let $\mathfrak{m}_g$ denote a unicellular map of genus $g$ having $n$ edges.
Let furthermore $\tau^I$ denote a trisection of type I and $\tau^{II}$ denote a
trisection of type II. Then we have the mappings $\Phi$ and $\Psi$:
$$
\Phi(\mathfrak{m}_g ,v_1,v_2,v_3)=(\mathfrak{m}_{g+1}, \tau^{I}), \quad
\Psi(\mathfrak{m}_g ,v_1,v_2,\tau)=(\mathfrak{m}_{g+1}, \tau^{II})
$$
are bijections, where $v_1$, $v_2$ and $v_3$ denote three distinct vertices in
$\mathfrak{m}_g$ and $\mathfrak{m}_{g+1}$ is a unicellular map of genus $g+1$
having $n$ edges.
\end{proposition}

Here $\Phi$ generates the trisection $\tau^{I}$ in a unicellular map of genus $g+1$
and the trisection $\tau^{II}$ persists when applying the mapping $\Psi$.

Gluing can be described as follow: \\
Given a unicellular map of $\mathfrak{m}_{g-k}$, together with a sequence of
vertices $V=\{v_1, \ldots v_{2k+1}\}$, where $v_i<_{\gamma} v_{i+1}$, $\forall 1\le i <2k+1$.
Then: \\
{\bf I.} we glue the last three vertices $v_{2k-1}$, $v_{2k}$ and $v_{2k+1}$ via $\Phi$, thereby
obtaining the unicellular map $\mathfrak{m}_{g-k+1}$ together with trisection $\tau^I$. \\
{\bf II.} we apply $\Psi(\mathfrak{m}_{g-k+i}, v_{2k-2i-1}, v_{2k-2i}, \tau^I)$ $k-1$ times for
$i=1$ to $i=k-1$. This produces the unicellular map $\mathfrak{m}_g(n)$, together with a
trisection $\tau^{II}$.
The process defines a mapping
$$
\Lambda(\mathfrak{m}_{g-k}, v_1, \ldots, v_{2k+1})=(\mathfrak{m}_g, \tau),
$$
where we do not label $\tau$ by type since in general we do not know whether $\Psi$ has been
applied. The order of the vertices in $V$ is given by the partial order determined by
$\gamma$. Thus $V$ can be considered as a set of vertices in $\mathfrak{m}_{g-k}$, ordered
by $<{\gamma}$. $\Lambda$ merges vertices from right to left by first
applying $\Phi$ once then applying $\Psi$ several times.

$\Lambda$ is reversed as follows: given a unicellular map $\mathfrak{m}_g$ of genus $g$
and $i=0$: \\
{\bf 1.} if $\tau$ is type II trisection in $\mathfrak{m}_{g-i}$, then let $(\mathfrak{m}_{g-i-1},
v_{2i+1},v_{2i+2}, \tau)=\Psi^{-1}(\mathfrak{m}_{g-i},\tau)$. We increase $i$ to $i+1$ and repeat step {\bf 1}. \\
{\bf 3.} if $\tau$ has type I, let $(\mathfrak{m}_{g-i}, v_{2i+1},v_{2i+2}, v_{2i+3})=\Phi^{-1}(
         \mathfrak{m}_{g-i-1},\tau)$. \\
Then we return
$$
\Xi(\mathfrak{m}_g ,\tau)=(\mathfrak{m}_{g-i}, V_{\tau}).
$$
By construction, $\Lambda$ and $\Xi$ are inverse to each other.

\begin{theorem} \cite{Chapuy:11} \label{T:bij}
Let $U_g^t$ denote the set of tuples $(\mathfrak{m}_g, v_1, \ldots, v_{t})$,
where $v_1, \ldots, v_{t}$ is a sequence of vertices in $\mathfrak{m}_g$.
Furthermore, let $D_g$ denote the set of tuples $(\mathfrak{m}_g, \tau)$,
where $\tau$ is a trisection of $\mathfrak{m}_g$.
Then
$$
\Lambda\colon {\dot\bigcup}_{k=0}^{g-1} U_k^{2g-2k+1} \rightarrow D_g,
\quad
\Xi\colon D_g \rightarrow {\dot\bigcup}_{k=0}^{g-1} U_k^{2g-2k+1}
$$
are bijections and $\Lambda\circ \Xi=\text{\rm id}$ and $\Xi\circ \Lambda=\text{\rm id}$.
\end{theorem}

Let $\epsilon_g(n)$ denote the number of unicellular map of genus $g$ having $n$ edges.
Then we have the following enumerative corollary
\begin{corollary}
\begin{equation} \label{E:induction1}
2g\cdot \epsilon_g(n) = {n+1-2(g-1) \choose 3} \epsilon_{g-1}(n)
+ \cdots + {n+1 \choose 2g+1} \epsilon_0(n).
\end{equation}
\end{corollary}
Here the $2g$-factor on left hand side counts the number of trisection in $\mathfrak{m}_g$ and
the binomial coefficients on the right hand side counts the number of distinct
selections of subsets of $(2k+1)$ vertices from a unicellular map $\mathfrak{m}_{g-k}$.

Iterating $\Xi$, we obtain
\begin{equation} \label{E:induction2}
\epsilon_g(n)  = \sum_{0=g_0<g_1<\cdots<g_r=g} \prod_{i=1}^r
\frac{1}{2g_i}{n+1-2g_{i-1} \choose 2(g_i-g_{i-1})+1} \cdot \epsilon_0(n),
\end{equation}
where $\epsilon_0(n)$ is the number of planar trees having $n$ edges, i.e.~the Catalan
number $\frac{1}{n+1}{2n \choose n}$.

\section{Uniform generation of matchings}\label{S:UniformM}

In this section, we show how to generate a matching of a given genus $g$
over $2n$ vertices with uniform probability.

Any unicellular map $\mathfrak{m}_g$ together with one of its $2g$ trisections
is mapped via $\Xi$ into a unicellular map of lower genus. Note that the genus decreases
at least by one. Therefore, by iterating the process finitely many times (at most $g$),
we arrive at a unicellular map of genus $0$, i.e~a planar tree.


For our construction it is important to keep track of the particular slicing process.
Accordingly, we introduce {\it slice/glue paths} as follows.
\begin{definition}
Suppose $\mathfrak{m}_g$ is a unicellular map of genus $g$ having $n$ edges.
Then a sequence unicellular maps
$$
(\mathfrak{m}^0=\mathfrak{m}_{g_0=0}, \mathfrak{m}^1=\mathfrak{m}_{g_1}, \ldots,
\mathfrak{m}^r=\mathfrak{m}_{g_r=g} )
$$
is called a slice path from $\mathfrak{m}_g$ to $\mathfrak{m}_0$
and a glue path when considered from $\mathfrak{m}_0$ to $\mathfrak{m}_g$,
where $\Xi(\mathfrak{m}_{g_i}, \tau_i)=(\mathfrak{m}_{g_{i-1}}, V_{g_{i-1}})$
holds for some $\tau_i$ in $\mathfrak{m}_{g_i}$, $0< i \le r$.
\end{definition}

We next consider $P_g(\mathfrak{m}^0)$, the set of distinct glue paths
from a given $\mathfrak{m}^0=\mathfrak{m}_0$ to some unicellular maps of
fixed genus $g$.

\begin{lemma}
The cardinality of $P_g(\mathfrak{m}^0)$ is given by
$$
\sum_{0=g_0<g_1<\cdots<g_r=g}\prod_{i=1}^r \frac{1}{2g_i}{n+1-2g_{i-1} \choose 2(g_i-g_{i-1})+1}.
$$
\end{lemma}
\begin{proof}
In order to construct $\mathfrak{m}_{g_{i}}$ from $\mathfrak{m}_{g_{i-1}}$,
$0<i\le r$, we need to select $2(g_{i}-g_{i-1})+1$ vertices from
$\mathfrak{m}_{g_{i-1}}$.

Euler characteristic shows that there are $(n+1-2g_{i-1})$ distinct vertices in
$\mathfrak{m}_{g_{i-1}}$, whence there are
${n+1-2g_{i-1} \choose 2(g_{i}-g_{i-1})+1}$ ways to select a subset of vertices
$V_{g_{i-1}}$.

On the other hand, the mapping $\Lambda$ produces $\mathfrak{m}_{g_i}$ with a
labeled trisection $\tau_i$, i.e., the same $\mathfrak{m}_{g_i}$ will be produced
exactly $2g_i$ times.
Accordingly, we need to normalize the production by a factor $1/2g_i$ for each
application of $\Lambda$.

As a result the total number of glue paths in $P_g(\mathfrak{m}^0)$ is
$$
\sum_{0=g_0<g_1<\cdots<g_r=g}\prod_{i=1}^r \frac{1}{2g_i}{n+1-2g_{i-1} \choose 2(g_i-g_{i-1})+1},
$$
which is exactly the coefficient of $\epsilon_0(n)$ in eq.~(\ref{E:induction2}).
\end{proof}

The problem of generating a unicellular map of genus $g$ having $n$ edges with
uniform probability thus splits into two parts: we first generate a planar
tree $\mathfrak{m}_0$ with $n$ edges with uniform probability. Second we
generate a glue path from $P_g(\mathfrak{m}^0)$ with uniform probability.
It is well-known how to implement the first step by a linear time (rejection) sampler \cite{SecondaryStructureSampling}
and it thus remains to present an algorithm for the second step.

We construct a glue path inductively.
Suppose we are at step $i$ and we have constructed a unicellular map
$\mathfrak{m}^i$ of genus $g_i$. Then the next
genus $g_{i+1}$ is suggested by the process ${\tt NextGenus}$. This process
considers the sequence of genus $g_0,\ldots, g_i$ and the target genus $g$
as input, and returns the genus $g_{i+1}$.
Let $\mathbb{P}(g_{i+1}=t\mid g_0, \ldots, g_i, g)$ denote the probability of
$g_{i+1}$ equals $t$ under the condition that $g_0, \ldots, g_i$ are
the genus of the previous steps and $g$ is the target genus.
Then
\begin{equation} \label{E:nextgenus}
\mathbb{P}(g_{i+1}=t\mid g_0, \ldots, g_i, g)=
\frac{\sum\limits_{t_0=g_0, \ldots, t_i=g_i, g_{i+1}=t<t_{i+1}<\cdots<t_r=g}
\prod\limits_{i=1}^r \frac{1}{2t_i} {n+1-2t_{i-1} \choose 2(t_i-t_{i-1})+1}}
{\sum\limits_{t_0=g_0, \ldots, t_i=g_i<\cdots<t_r=g}\prod\limits_{i=1}^r \frac{1}{2t_i}{n+1-2t_{i-1} \choose 2(t_i-t_{i-1})+1}}.
\end{equation}

Next we select the sequence of vertices from $\mathfrak{m}^{i}$ by process
{\tt SelectVertex}.
This process chooses vertices in $2(g_{i+1}-g_{i})+1$ independent steps. The probability
of a vertex being selected is given by
$1/(n+2-2g_{i}-k)$, where $(n+2-2g_{i}-k)$ is the number of remaining non-selected
vertices in the $k$th step, $1\le k \le 2(g_{i+1}-g_{i})+1$.
Since the selected vertices are ordered automatically by $<_{\gamma_{\mathfrak{m}^i}}$,
the same set is generated with multiplicity $(2(g_{i+1}-g_{i})+1)!$.
Normalizing the resulting term by the factor $1/(2(g_{i+1}-g_{i})+1)!$,
the probability of the set $V_{i}$, $0\le i< r$ is given by
\begin{equation} \label{E:select}
\mathbb{P}_{\tt Select}(V_{i}) = \frac{1}{(2(g_{i+1}-g_{i})+1)!} \cdot
\frac{1}{n+1-2g_{i}} \cdots  \frac{1}{n-2g_{i+1}}
= \frac{1}{{n+1-2g_{i} \choose 2(g_{i+1}-g_{i})+1}}.
\end{equation}

After the sequence of vertices $V_i$ is selected, a unicellular map $\mathfrak{m}^{i+1}$ is
constructed by the process {\tt Glue}, applying mapping $\Lambda$.
We present the pseudocode of the procedures in Algorithm~\ref{A:path}.

\begin{algorithm}
\begin{algorithmic}[1]
\STATE {\tt UniformMatching}~($\mathfrak{m}^0, TargetGenus$)
\STATE {$i\leftarrow 0$}
\WHILE {$g_i \le TargetGenus$}
\STATE {$g_{i+1} \leftarrow {\tt NextGenus}~(g_0, \ldots, g_i, TargetGenus)$}
\STATE {$V_i \leftarrow {\tt SelectVertex}~(\mathfrak{m}^i, 2(g_{i+1}-g_i)+1)$}
\STATE {$\mathfrak{m}^{i+1} \leftarrow {\tt Glue}~(\mathfrak{m}^i, V_i)$}
\STATE {$i\leftarrow i+1$}
\ENDWHILE
\STATE \textbf{return} $\mathfrak{m}^i$
\caption {\small }
\label{A:path}
\end{algorithmic}
\end{algorithm}
Assuming the target genus to be constant and taking into account that during our construction the
genus is strictly increasing, the while-loop of Algorithm 1 is executed only a constant number of times.
Using appropriate memorization techniques, {\tt NextGenus} and {\tt Glue} can
be implemented in constant time and {\tt SelectVertex} in linear time.
Thus, combined with a linear time sampler
for planar trees, our approach allows for the uniform generation of random matchings in time $O(n)$.
\begin{lemma}
Given a planar tree $\mathfrak{m}^0$ with $n$ edges and a genus $g$,
the probability of a glue path $p_g$ generated by Algorithm~\ref{A:path} is
$\epsilon_0(n)/\epsilon_g(n)$.
\end{lemma}
\begin{proof}
Assume a glue path
$$
p_g=\{\mathfrak{m}^0=\mathfrak{m}_{g_0=0}, \mathfrak{m}^1=\mathfrak{m}_{g_1}, \ldots,
\mathfrak{m}^r=\mathfrak{m}_{g_r=g} \}
$$
is generated by Algorithm~\ref{A:path}.
Since for each step, the process of choosing the genus for the next step and selecting labeled vertices
is independent, the probability of $P_{g}$ is given by
$$
\mathbb{P}(P_g) = \prod_{i=0}^{r-1} \mathbb{P}(g_{i+1}=t_{i+1}|g_0,\ldots, g_i, g) \cdot \mathbb{P}(V_i).
$$
We substitute eq.~(\ref{E:nextgenus}) and eq.~(\ref{E:select}) and obtain
\begin{eqnarray*}
\mathbb{P}(P_g) & = & \prod_{i=0}^{r-1} \frac{\sum_{t_0=g_0, \ldots, t_i=g_i, g_{i+1}=t<t_{i+1}<\cdots<t_r=g}
\prod_{i=1}^r \frac{1}{2t_i} {n+1-2t_{i-1} \choose 2(t_i-t_{i-1})+1}}
{\sum_{t_0=g_0, \ldots, t_i=g_i<\cdots<t_r=g}\prod_{i=1}^r \frac{1}{2t_i}{n+1-2t_{i-1} \choose 2(t_i-t_{i-1})+1}}
\cdot \frac{1}{{n+1-2g_{i} \choose 2(g_{i+1}-g_{i})+1}} \\
& = & \frac{\prod_{t_1=g_1, \ldots, t_r=g_r} \frac{1}{2t_i}{n+1-2t_{i-1} \choose 2(t_i-t_{i-1})+1}}
{\sum_{0=t_0<\cdots<t_r=g} \prod_{i=1}^r \frac{1}{2t_i}{n+1-2t_{i-1} \choose 2(t_i-t_{i-1})+1}}
\cdot  \prod_{i=0}^{r-1} \frac{1}{{n+1-2g_{i} \choose 2(g_{i+1}-g_{i})+1}} \\
& = & \frac{1}{\sum_{0=t_0<\cdots<t_r=g} \prod_{i=1}^r \frac{1}{2t_i}{n+1-2t_{i-1} \choose 2(t_i-t_{i-1})+1}} \\
& = & \frac{\epsilon_0(n)}{\epsilon_0(n)\cdot
\sum_{0=t_0<\cdots<t_r=g} \prod_{i=1}^r \frac{1}{2t_i}{n+1-2t_{i-1} \choose 2(t_i-t_{i-1})+1}} \\
& = & \frac{\epsilon_0(n)}{\epsilon_g(n)},
\end{eqnarray*}
whence the lemma.
\end{proof}

\begin{corollary}
Suppose a planar tree $\mathfrak{m}^0$ is uniformly generated, i.e., with probability
$1/\epsilon_0(n)$. Then a unicellular map $\mathfrak{m}^r=\mathfrak{m}_g$ is uniformly
generated by Algorithm~\ref{A:path} with probability $1/\epsilon_g(n)$.
\end{corollary}

\section{Uniform generation of diagrams} \label{S:UniformD}
In this section, we extend our result of Section~\ref{S:UniformM} in order
to generate diagrams of genus $g$ with uniform probability. The idea is to
uniformly generate first a matching of genus $g$ with $n$ arcs. In a second
step we choose $(\ell-2n)$ unpaired vertices and insert them into
the matching.

Let
$\mathbb{P}_{d}(t=n|\ell, g)$ denote the probability of the diagram having
exactly $n$ arcs, $0\le n \le \lfloor \ell/2\rfloor$. In the following we
compute $\mathbb{P}_{d}(t=n|\ell, g)$.

Let $\delta_g(\ell)$ denote the number of diagrams of genus $g$ over $\ell$
vertices. Furthermore, let $\delta_g(\ell ,n)$ denote the number of diagrams
of genus $g$ over $\ell$ vertices having exactly $n$ arcs, $2n\le \ell$.
Then $\ell-2n$ vertices are unpaired and
$$
\delta_g(\ell, n) = {\ell \choose \ell-2n} \epsilon_g(n).
$$
Furthermore
\begin{equation}
\delta_g(\ell)=\sum_{n=0}^{\lfloor l/2\rfloor} \delta_g(\ell, n)
=\sum_{n=0}^{\lfloor l/2\rfloor} {\ell \choose \ell-2n} \epsilon_g(n).
\end{equation}
In order to generate a diagram of genus $g$ over $\ell$ arcs uniformly,
we need to solve
$$
\frac{1}{\delta_g(\ell)} = \mathbb{P}_{d}(t=n|\ell, g) \cdot \frac{1}{\epsilon_g(n)} \cdot
\frac{1}{{\ell \choose \ell-2n}},
$$
whence $\mathbb{P}_{d}(t=n|\ell, g)=\delta_g(\ell, n)/\delta_g(\ell)$.

We present the pseudocode of {\tt UniformDiagram} as
Algorithm~\ref{A:diagram}. The subroutine
{\tt NumberofArcs} returns $n$ with probability
$\mathbb{P}_{d}(t=n|\ell, g)$, which determines the number of arcs in diagram $D_g(\ell)$.
{\tt UnifomTree} is a standard process uniformly generating a matching of genus $0$ with
$n$ arcs. Finally, the process {\tt InsertUnpairedVertices}
first chooses $(\ell-2n)$ vertices from $\ell$ vertices as unpaired. It leaves
$2n$ vertices not selected, which are considered to be paired. Then the process
maps the $2n$ vertices of the matching generated by {\tt UniformMatching} and keeps
the arcs in the upper half-plane. Accordingly, a diagram of genus $g$ over $\ell$ vertices
with exactly $n$ arcs is generated. The result of some experiments conducted in connection with
the generation of random matchings and diagrams using our algorithms is shown in Figure~\ref{F:plot}.

\begin{algorithm}
\begin{algorithmic}[1]
\STATE {\tt UniformDiagram}~($\ell, TargetGenus$)
\STATE {$n\leftarrow {\tt NumberofArcs}(\ell, g)$}
\STATE {$\mathfrak{m}_0 \leftarrow {\tt UnifomTree}(n)$}
\STATE {$\mathfrak{m}_g \leftarrow {\tt UniformMatching}(\mathfrak{m}, TargetGenus)$}
\STATE {$D_g \leftarrow {\tt InsertUnpairedVertices}(\mathfrak{m}_g, \ell)$}
\STATE \textbf{return} $D_g$
\caption {\small }
\label{A:diagram}
\end{algorithmic}
\end{algorithm}

\begin{figure}[ht]
\begin{center}
\includegraphics[width=0.8\columnwidth]{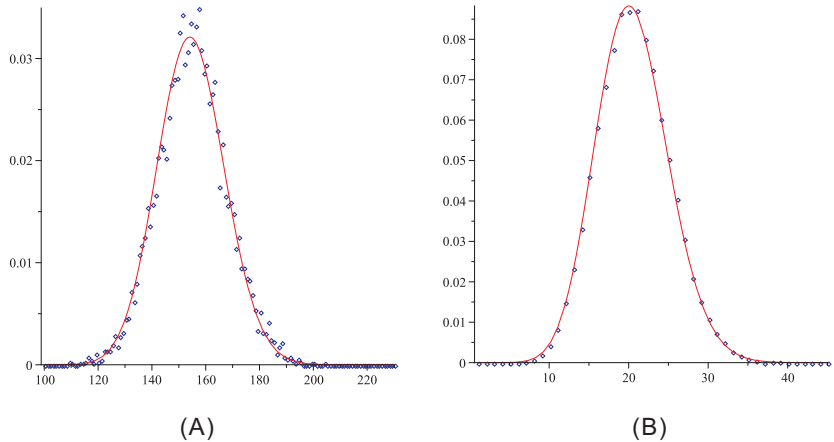}
\end{center}
\caption{\small Uniform generation: (A) matchings, $n=12$, $g=2$ and
$\epsilon_2(n/2)=6468$. We generate $N=10^6$ matchings and display the frequencies
of their multiplicities (blue dots) together with the binomial coefficient of the
uniform sampling
${N \choose \ell}(1/\epsilon_2(n/2))^\ell(1-1/\epsilon_2(n/2))^{N-\ell}$ (red).
(B) The analog of (A) for diagrams. Here we have $n=12$, $g=2$ and $\delta_2(n)=48741$.
We generate $N=10^6$ diagrams and display the frequencies of their multiplicities
(blue dots) together with the binomial coefficients
${N \choose \ell}(1/\delta_2(n))^\ell(1-1/\delta_2(n))^{N-\ell}$ (red).
}
\label{F:plot}
\end{figure}

\section{Non-uniform sampling}

RNA structures can be represented as diagrams and are, due to the biophysical
context subject to certain constraints with respect to their free energy \cite{Mathews:99}.
The latter energy is oftentimes modeled as a function of the loops of the
underlying RNA structure \cite{Mathews:99}, $S$. These loops are in fact equal
to the boundary components of the fatgraph constructed from the molecule.
In the following we shall discuss, $\eta(S)$, a simplified version of the actual
bio-physical loop-energy of a structure $S$.

Let us start with RNA secondary structures, that correspond to diagrams of genus $0$.
For a secondary structure $S_0$, we denote its corresponding (see Section~\ref{S:basic},
duality mapping $\pi$) unicellular map by
$\mathfrak{m}_0=\pi(S_0)$. The bonds or arcs of the structure then correspond to edges of
the unicellular map $\mathfrak{m}_0$ and loops or boundary components to vertices.
Three types of loops are distinguished: hairpin loops, interior loops
(including helices and bludge loops) and multi-loops.
Accordingly, the duality maps hairpin loops into vertices of degree one,
interior loops into vertices of degree two, and multi-loop into vertices of
degree greater than two, see Figure~\ref{F:loops}.
$\eta(S)$ extends these types in order to deal with structures having
arbitrary genus $g\ge 0$ as follows.

\begin{figure}[ht]
\begin{center}
\includegraphics[width=0.8\columnwidth]{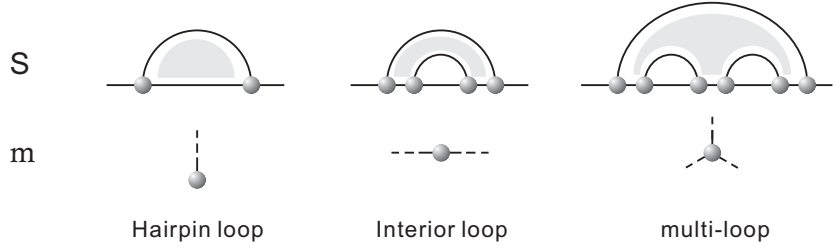}
\end{center}
\caption{\small Hairpin loops, interior loops and multi-loops in
secondary structures and their translation into unicellular maps.
}
\label{F:loops}
\end{figure}

Let $S_g$ denote an RNA structure having length $\ell$, $n$ arcs
and genus $g$, and $\mathfrak{m}_g=\pi(S_g)$ its corresponding
unicellular map. Then $\eta(S_g)$ is given by
\begin{equation}
\eta(S_g) = n\cdot b +\sum_{v\in \mathbb{V}} T(v) + L^{pk}_g.
\end{equation}
Here $b$ represents an energy contribution of arcs, $\mathbb{V}$ the set
of all vertices $\mathfrak{m}_g$, $T(v)$ is function given by
$$
T(v) =
\begin{cases}
L^{hp} \quad \text{if $d(v)=1$}, \\
L^{int} \quad \text{if $d(v)=2$},  \\
L^{mul} \quad \text{if $d(v)> 2$ and $v$ has no trisection,} \\
0 \quad \text{if $v$ is attached to the root,}
\end{cases}
$$
where $d(v)$ is the degree of vertex $v$, and $L^X$ is the contribution of a loop of
type $X$, where $X=\{hp, int, mul\}$.
Finally, $L^{pk}_g$ represents a contribution that stems from novel loop-types
emerging for genus $g>0$. In this model, we do not take contributions from unpaired
vertices into account.

In case of $g=1$, there are four different types of pseudoknots \cite{Huang:10a},
see Figure~\ref{F:reduction}. This is analogous for any genus: there are always
only finitely many corresponding shadows \cite{Huang:09,Huang:10a}, see Figure~\ref{F:reduction}.
Here, a shadow is a diagram without unpaired vertices in which all stacks (parallel arcs)
have size one.
Formally, a shadow of a structure can be obtained by first removing all its
unpaired vertices, second removing all noncrossing arcs (together with their vertices) and then replacing
a set of parallel arcs (and the incident vertices) of the form $\{(i,j), (i+1, j-1), \ldots, (i+\ell-1, j-\ell+1)\}$
by a single arc (and two vertices).

Let us have a closer look at the boundary components of these shadows in case of genus $1$.
(H) is inspected to have one boundary component, whence $\eta(S_\text{H})=2b+L^{mul}+L^{pk}_1$.
(K) and (L) have two boundary components and accordingly $\eta(S_\text{K})=\eta(S_\text{L})=3b+2L^{mul}+L^{pk}_1$.
Finally, for (M) we have $\eta(S_\text{M})=4b+3L^{mul}+L^{pk}_1$.

Consider a matching $S_1$ of genus $1$ having $n$ arcs and $\mathfrak{m}_1=\pi(S_1)$
the unicellular map given by the duality.
By selecting a trisection $\tau$ in $\mathfrak{m}_1$ and applying the mapping $\Xi$, we obtain
$\Xi(\mathfrak{m}_1, \tau)=(\mathfrak{m}_0, v_1, v_2, v_3)$
and three labeled vertices. Here we write $\mathfrak{m}_0^{(3)}=(\mathfrak{m}_0, v_1, v_2, v_3)$
for short. Let $S_0^{(3)}$ denote a
secondary structure with three labeled boundary component where
$\pi(S_0^{(3)})=\mathfrak{m}_0^{(3)}$. Let further
$\mathbb{S}_{0,n}^{(3)}$ and $\mathbb{S}_{1,n}$ denote the set of
$S_1$ and $S_0^{(3)}$ respectively.
By Lemma~\ref{L:trisection}, there are two trisections in $\mathfrak{m}_1$.
Therefore, by selecting the same $S_1$ and different trisection $\tau$,
$\Xi$ results in different $S_0^{(3)}$. Thus
we have the cardinality $2|\mathbb{S}_{1,n}| = |\mathbb{S}_{0,n}^{(3)}|$.
Figure~\ref{F:reduction} shows this for the four shadows of genus $1$
and their secondary structure with three labeled boundary components.

\begin{figure}[ht]
\begin{center}
\includegraphics[width=0.8\columnwidth]{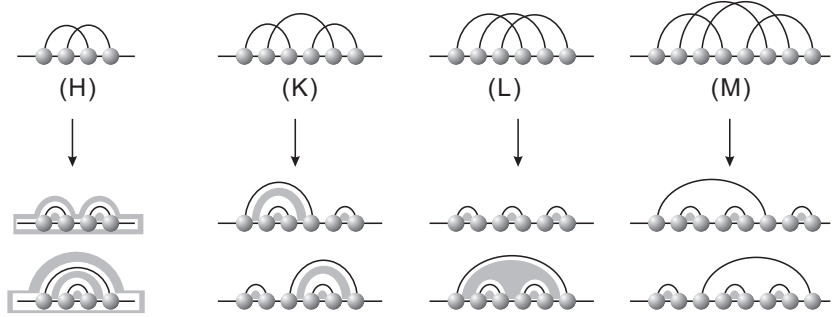}
\end{center}
\caption{\small Correspondence between shadows of genus one and
their labeled secondary structures.
}
\label{F:reduction}
\end{figure}

We next formulate an ``energy'' for structures $S_0^{(3)}$, $\eta(S_0^{(3)})$,  that matches the
energy $\eta$ for their corresponding counterpart of genus one after gluing.
Note that this allows us to reduce everything to secondary structures with three labeled boundary
components.
To this end, let $v_1$, $v_2$ and $v_3$ be three labeled vertices in $\mathfrak{m}_0^{(3)}$, where
$\mathfrak{m}_0^{(3)}=\pi(S_0^{(3)})$. Setting $T(v_1)=T(v_2)=T(v_3)=(L_1^{pk}+L^{mul})/3$ we observe

\begin{proposition} \label{P:score}
We have $\eta(S_1) = \eta(S_0^{(3)})$.
\end{proposition}

\begin{proof}
The mapping $\Xi$ is a bijection and $\Lambda(\mathfrak{m}_0, v_1, v_2, v_3)=
(\mathfrak{m}_1,
\tau)$. The three labeled vertices in $\mathfrak{m}_0$ are glued as $\overline{v}$, where
$d(\overline{v})\ge 3$. Hence $T(\overline{v})=L^{mul}+L^{pk}_1=T(v_1)+T(v_2)+T(v_3)$, because
$T(v_1)=T(v_2)=T(v_3)=(L^{pk}_1+L^{mul})/3$. The other vertices in $\mathfrak{m}_0$ maintain
hence their scores are not changed.
\end{proof}

Given $\eta(S)$ we proceed along the lines of \cite{McCaskill:90} and construct
a probability space of structures by computing the partition function of a given
sequence. Let $\theta(n)=\sum_{S\in \mathbb{S}_{n}} e^{\eta(S)}$ denote the
total energy of all structures. A structure, $S$, is sampled with probability
$e^{\eta(S)}/\theta(n)$.
In case of secondary structures, loop-based and arc-based energy models are
compatible to the standard recursion of secondary structure \cite{Waterman:78a}
and $\theta_0(n)$ can be computed by the recursion
$$
\theta(n) = \sum_{i=1}^{n-2} \theta(i)\theta(n-i-1)
+ e^{{L^{hp}}+b} \cdot \theta(1) + e^{{L^{int}}+b} \cdot \theta(n-1)
+ e^{{L^{mul}}+b} \cdot \sum_{i=1}^{n-3} \theta(i)\theta(n-i-2),
$$
where $\eta$ is the energy functional, discussed above. As there is only one summation
in the above recursion, $\theta(n)$ is computed in $O(n^2)$ time.

We proceed by showing that the new functional, $\eta(S_0^{(3)})$, {is} also compatible
with the secondary structure recursions.

\begin{lemma}
Let $\theta_1(n)=\sum_{S\in \mathbb{S}_{1,n}} e^{\eta(S)}$ and
$\theta_0^{(3)}(n)=\sum_{S\in \mathbb{S}_{0,n}^{(3)}} e^{\eta(S)}$.
Then $\theta_1(n)=\theta_0^{(3)}(n)/2$ can be computed in $O(n^2)$ time.
Once $\theta_1(n)$ is computed, a structure of genus one, $S_1$, is sampled
with probability $e^{\eta(S_1)}/\theta_1(n)$ in $O(n)$ time.
\end{lemma}
\begin{proof}
We have $\eta(S_1)=\eta(S_0^{(3)})$ for all $S_1\in \mathbb{S}_{1,n}$ and
$S_0^{(3)}\in \mathbb{S}_{0,n}^{(3)}$, and $2|\mathbb{S}_{1,n}|=|\mathbb{S}_{0,n}^{(3)}|$.
Therefore,
$$
\theta_1(n)=\sum_{S\in \mathbb{S}_{1,n}} e^{\eta(S)} = \frac{1}{2}
\sum_{S\in \mathbb{S}_{0,n}^{(3)}} e^{\eta(S)} = \frac{1}{2} \theta_0^{(3)}(n).
$$
We next show that $\theta_0^{(3)}(n)$ can be computed in $O(n^2)$ time.
Let $\theta_0^{(2)}= \sum_{S\in \mathbb{S}_{0,n}^{(2)}} e^{\eta(S)}$ and
$\theta_0^{(1)}=\sum_{S\in \mathbb{S}_{0,n}^{(1)}} e^{\eta(S)}$, where
$\mathbb{S}_{0,n}^{(2)}$ and $\mathbb{S}_{0,n}^{(1)}$ denote the sets of
secondary structures with two and one labeled boundary components.
The functionals of these labeled boundary components are computed exactly
as in the case of $S_0^{(3)}$.

Then we have, see also Figure~\ref{F:theta_recursion}:
\begin{eqnarray*}
\theta_0^{(3)}(n)& =& 2\sum_{i=1}^{n-2}\theta_0^{(3)}(i)\theta_0(n-i-1) +
2\sum_{i=1}^{n-2}\theta_0^{(2)}(i)\theta_0^{(1)}(n-i-1) \\
& + & e^{(L^{mul}+L_1^{pk})/3+b}\cdot \left(2\sum_{i=1}^{n-2}\theta_0^{(2)}(i)\theta_0(n-i-1) +
\sum_{i=1}^{n-2}\theta_0^{(1)}(i)\theta_0^{(1)}(n-i-1) \right) \\
& + & 2e^{L^{mul}+b}\cdot \left( \sum_{i=1}^{n-3}\theta_0^{(3)}(i)\theta_0(n-i-2) +
\sum_{i=1}^{n-3}\theta_0^{(2)}(i)\theta_0^{(1)}(n-i-2) \right) \\
& + & e^{(L^{mul}+L_1^{pk})/3+b}\cdot \left(2\sum_{i=1}^{n-3}\theta_0^{(2)}(i)\theta_0(n-i-2) +
\sum_{i=1}^{n-3}\theta_0^{(1)}(i)\theta_0^{(1)}(n-i-2) \right) \\
& + & e^{L^{int}} \cdot  \theta_0^{(3)}(n-1) + e^{(L^{pk}_1+L^{mul})/3+b} \cdot  \theta_0^{(2)}(n-1).
\end{eqnarray*}

\begin{figure}[ht]
\begin{center}
\includegraphics[width=0.9\columnwidth]{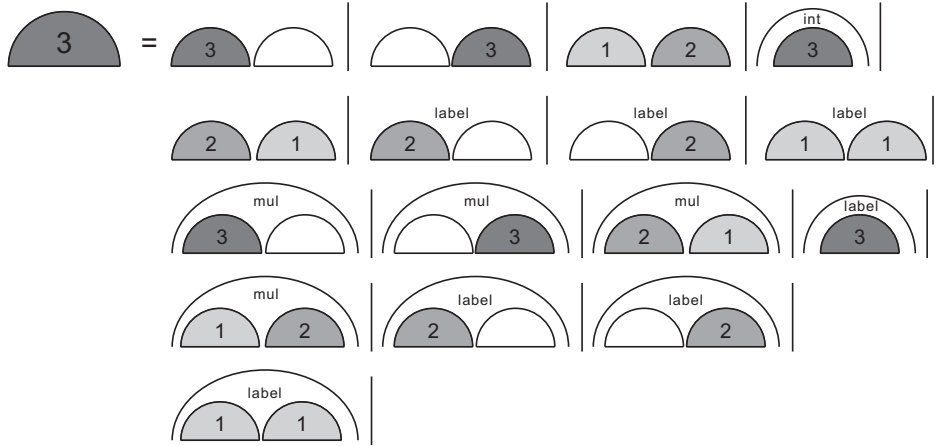}
\end{center}
\caption{\small The recursion for $\theta_0^{(3)}$.
}
\label{F:theta_recursion}
\end{figure}

Analogously, we have recursions for $\theta_0^{(2)}(n)$ and $\theta_0^{(1)}(n)$,
which can be computed in $O(n^2)$ time. Therefore, $\theta_1(n)=\theta_0^{(3)}(n)$
can be computed in $O(n^2)$ time.

In order to sample a diagram of genus $1$ over $\ell$ vertices, $D_1$, we
need first to determine its number of arcs. As in the case of uniform sampling,
we have $\vartheta_1(\ell)=\sum_{n=0}^{\lfloor \ell/2\rfloor} \vartheta_1(\ell,n)$,
where $\vartheta_1(\ell,n)={\ell \choose \ell-2n} \theta_1(n)$.
Replacing in the formulae for uniform sampling
$\epsilon_1(n)$ by $\theta_1(n)$ and $\delta_1(n)$ by $\vartheta_1(n)$, we find
that the probability of sampling a diagram with $n$ arcs is given by
$\vartheta_1(\ell, n)/\vartheta_1(\ell)$.
It remains to sample a matching $S_0^{(3)}$ with $n$ arcs and to subsequently
glue the three labeled vertices in $\mathfrak{m}_0^{(3)}=\pi (S_0^{(3)})$.
This generates a unicellular maps of genus one, $\mathfrak{m}_1$, which is
associated to $S_1=\pi^{-1}(\mathfrak{m}_1)$ by duality. Note that choosing different
slice-paths for $S_1$ generates two different $S_0^{(3)}$, see eq.~(\ref{E:commute}).
\begin{equation}\label{E:commute}
\diagram
S_0^{(3)} \dto \rto & S_1 \\
(\mathfrak{m}_0, v_1, v_2, v_3) \rto^{\Lambda} & (\mathfrak{m}_1,\tau) \uto
\enddiagram
\end{equation}
The probability of a structure of genus one, $S_1$, is then given by
$$
\frac{2 e^{\eta(S_0^{(3)})}}{\theta_0^{(3)}(n)} = \frac{2e^{\eta(S_1)}}{2\theta_1(n)}
= \frac{e^{\eta(S_1)}}{\theta_1(n)}.
$$
Finally, we insert the unpaired vertices into $S_1$ and obtain $D_1$
with the probability
$$
\frac{\vartheta_1(\ell, n)}{\vartheta_1(\ell)} \cdot
\frac {{\ell \choose \ell-2n}e^{\eta(D_1)}}{\theta_1(n)}
=\frac{e^{\eta(D_1)}}{\vartheta_1(\ell)}.
$$
\end{proof}

\section{Conclusion}

In this paper we have proposed an original and highly efficient (linear time) approach to sample random RNA
pseudoknotted structures in the uniform and a non-uniform model. The later builds on a simplified
concept of free energy, favoring foldings of a native appearance. This is a first step towards
efficient prediction algorithms for pseudoknotted RNA since structure predictions of good quality can easily be derived
from suitable (high quality) random samples (see \cite{NebelScheid:11} and the references given there).
To this end, our algorithms need to be extended towards two directions:
\begin{enumerate}
  \item The probability model needs to be improved further, and
  \item the RNA sequence needs to be taken into account.
\end{enumerate}

An immediate application of the uniform sampler are the distributions
of loops in structures of genus $g$. We have shown that the loops in
structures are translated  into vertices of their associated unicellular
maps. In particular, a hairpin loop corresponds to a vertex of degree one,
an interior loop to a vertex of degree two and a multi-loop is to some vertex
having degree greater than two without a trisection.
Finally a pseudoknot loop corresponds to a vertex having degree greater
than two containing a trisection. In Fig.~\ref{F:loop_statistic} we present the
respective data, filtered by genus.
\begin{figure}[ht]
\begin{center}
\includegraphics[width=0.9\columnwidth]{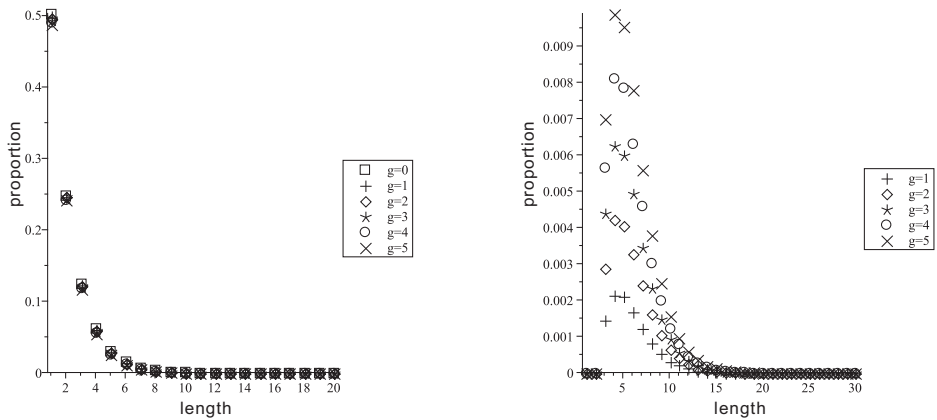}
\end{center}
\caption{\small Loops in uniformly generated, genus filtered, RNA structures.
Left: distribution of standard loops, where $x$-axis is the length
of boundary component and $y$-axis is frequency.
Right: distribution of pseudoknot loops.
}
\label{F:loop_statistic}
\end{figure}

It is well known in context of pseudoknot-free secondary structures how to use either a sophisticated
model for the free energy or stochastic concepts like the maximum likelihood approach to obtain
realistic probability models applicable to random sampling. Our approach seems to be suitable to
apply the latter and it is a topic for future research to work out the details. Incorporating the sequence
is a more complicated task but again results for classic RNA secondary structures prove it feasible with only
small losses in efficiency \cite{NebelScheid:12}.

Thus we assume our findings of this paper an important contribution towards
the development of efficient structure prediction tools for pseudoknotted RNA structures. Those are also in need
for state of the art tools addressing the inverse folding problem. The latter quite often use some search heuristic
(like e.g.\ a genetic algorithm) to process the space of possible sequences using structure prediction tools to
judge the quality (similarity to input) of current solutions. For the large number of calls, the efficiency of the
prediction algorithm is crucial for the applicability of the entire approach. Today's established algorithms
for the prediction of pseudoknotted RNA with run times in $O(n^4)$ or worth (see \cite{Nebel:12}) seem not to be appropriate.

\bibliographystyle{elsarticle-num}
\bibliography{gen}

\end{document}